\documentclass[10pt,conference]{IEEEtran}

\usepackage{times}
\usepackage{amsmath,dsfont}
\usepackage{amssymb,amsthm}
\usepackage{epsfig,verbatim}
\usepackage{mathrsfs}
\usepackage{verbatim}
\usepackage{ifthen}
\usepackage{syntonly}
\usepackage{graphicx}
\usepackage{epstopdf}
\usepackage{multirow}

\theoremstyle{remark} \newtheorem{theorem}{Theorem}
\theoremstyle{remark} 
\newtheorem{proposition}{Proposition}

\theoremstyle{remark} \newtheorem{remark}{Remark}

\newcommand{\eps}{\epsilon}

\newcommand{\styp}{A_{\epsilon}^{*(n)}}

\newcommand{\stypm}{A^{*(m)}_{\epsilon}}

\newcommand{\mX}{\mathcal{X}}

\newcommand{\mY}{\mathcal{Y}}
\newcommand{\mW}{\mathcal{W}}

\newcommand{\mS}{\mathcal{S}}
\newcommand{\mCN}{\mathcal{CN}}

\newcommand{\tU}{\tilde{U}}

\newcommand{\yvec}{\mathbf{y}}
\newcommand{\xvec}{\mathbf{x}}

\newcommand{\svec}{\mathbf{s}}

\newcommand{\vvec}{\mathbf{v}}
\newcommand{\wvec}{\mathbf{w}}

\newcommand{\tsvec}{\tilde{\svec}}

\newcommand{\msg}{u}

\newcommand{\E}{\mathds{E}}
\newcommand{\CN}{\mathcal{CN}}

\newcommand{\setR}{\mathfrak{R}}
\newcommand{\setC}{\mathfrak{C}}

\newcommand{\chR}{\hat{R}}

\long\def\symbolfootnote[#1]#2{\begingroup\def\thefootnote{\fnsymbol{footnote}}\footnote[#1]{#2}\endgroup}

\newcommand{\thmlabel}[1]{\label{thm:#1}}
\newcommand{\thmref}[1]{\ref{thm:#1}}
\newcommand{\Thmref}[1]{Thm.~\thmref{#1}}

 \IEEEoverridecommandlockouts
\title{Source-Channel Coding for the Multiple-Access Relay Channel
\thanks{
\noindent
This work was supported by the European Commission's Marie Curie IRG Fellowship PIRG05-GA-2009-246657  under the Seventh Framework Programme.}
}

\vspace{-0.5cm}

\author{
\IEEEauthorblockN{Yonathan Murin, Ron Dabora}
\authorblockA{Department of Electrical and Computer Engineering \\
Ben-Gurion University, Israel\\
Email: moriny@bgu.ac.il, ron@ee.bgu.ac.il}

\vspace{-1.1cm}

\and
\IEEEauthorblockN{Deniz G\"und\"uz}
\authorblockA{Centre Tecnologic de Telecomunicacions  \\
de Catalunya (CTTC), Barcelona, Spain \\
Email: deniz.gunduz@cttc.es}

\vspace{-1.1cm}

}

\vspace{-0.5cm}

\begin{document}

\maketitle
\thispagestyle{empty} \pagestyle{empty}

\begin{abstract}
	This work considers reliable transmission of general correlated sources over the multiple-access relay channel (MARC)
    and the multiple-access broadcast relay channel (MABRC). In MARCs only the destination is interested in a reconstruction of the sources,
    while in MABRCs both the relay and the destination want to reconstruct the sources.
	We assume that both the relay and the destination have correlated side information.
	We find sufficient conditions for reliable communication based on operational separation, as well as necessary conditions on the achievable source-channel rate.
	For correlated sources transmitted over fading Gaussian MARCs and MABRCs we find conditions under which informational separation is optimal.	
	
\end{abstract}

\vspace{-0.1cm}

\section{Introduction}

\vspace{-0.1cm}

The multiple-access relay channel (MARC) is a network in which several users communicate with a single destination with the
help of a relay \cite{Kramer:00}.
Examples of such a network are sensor and ad-hoc networks in which an intermediate relay can be added to assist communication
from the sources to the destination.
Achievable regions for the MARC were derived in \cite{Kramer:2005}, \cite{Sankar:07} and \cite{KramerMandayam:04}.
In \cite{Kramer:2005} Kramer et al. derived an achievable rate region for the MARC with independent sources.
The coding scheme employed in \cite{Kramer:2005} is based on decode-and-forward (DF) relaying, and uses regular encoding,
successive decoding at the relay and backward decoding at the destination.
In \cite{Sankar:07}, another DF-based coding scheme for the MARC was presented.
The work \cite{Sankar:07} also showed that, in contrast to DF for the classic relay channel, for the MARC different DF
schemes yield different rate regions (backward decoding can yield larger rates than sliding window decoding).
Outer bounds on the capacity of discrete memoryless (DM) MARCs were obtained in \cite{KramerMandayam:04}.

The previous work on MARCs considered independent messages at the terminals. In contrast, in the present work here we allow arbitrary
correlation among the sources that should be transmitted to the destination in a lossless fashion.

In \cite{Shannon:48} Shannon showed that a source can be reliably transmitted over a memoryless point-to-point (PtP) channel,
if and only if its entropy is less than the channel capacity. Hence, a simple comparison of the rates of the optimal
source and channel codes for the respective source and channel suffices to conclude if reliable communication is feasible.
This is called the separation theorem.
The implication of the separation theorem is that independent design of the source and channel codes is optimal.

In \cite{ShamaiVerdu:95} Shamai and Verdu considered the availability of correlated side information at
the receiver in a PtP scenario, and showed that source-channel separation is optimal.
The availability of receiver side information enables transmitting the source reliably over a channel with a smaller
capacity compared to the capacity needed in the absence of receiver side information.

Unfortunately, optimality of source-channel separation in the Shannon sense does not generalize to multiuser networks \cite{Shannon:61}, \cite{Cover:80}, \cite{GunduzErkip:09}.
Therefore, in general the source and channel codes must be jointly designed for every particular source and channel combination.
Source-channel coding over the broadcast channel was considered by Tuncel in \cite{Tuncel:06}.
Tuncel distinguishes between two types of source-channel separation.
\emph{Informational separation} refers to classical separation in the Shannon sense. \emph{Operational separation} refers to statistically independent source and channel codes that are not necessarily the optimal codes for the underlying source or the channel.
In \cite{Tuncel:06} Tuncel showed that for a broadcast channel in which each receiver has a different side information,
operational separation is optimal, while  informational separation is not.

Optimizing source coding along with multiuser channel coding in a general setting is a very complicated task.
In \cite{ChaoChenDiggaviShamai:Submitted} Tian et al. showed the optimality of operational separation for the following two scenarios: a) arbitrarily correlated sources over orthogonal links; b) independent sources over a general network with some restrictions on how many messages can be decoded at each destination.
In \cite{GunduzErkip:09} G\"und\"uz et al. obtained necessary and sufficient conditions for the optimality of
informational separation for the multiple-access channel with correlated sources and side information.
G\"und\"uz et al. also obtained necessary and sufficient conditions for the optimality of operational
separation for the compound multiple-access channel with correlated sources and side information.
In \cite{ErkipGunduz:07} G\"und\"uz and Erkip showed that operational separation is optimal for the cooperative
relay-broadcast channel.
Necessary and sufficient conditions for reliable transmission of a source over a relay channel when side information
is available either only at the receiver, or only at the relay or at both the relay and the
receiver were established in \cite{ElGamalCioffi:07}.

In this paper we shall also consider MARCs and MABRCs subject to independent and identically distributed (i.i.d.)
fading, for both phase and Rayleigh fading.
Phase fading models apply to high-speed microwave communications where the oscillator's phase noise and the system timing
jitter are the key impairments. Phase fading is also the major impairment in communication systems that employ orthogonal
frequency division multiplexing 
, as well as in some applications of naval communications.
Rayleigh fading models are very common in wireless communications and apply to mobile communications in the presence
of multiple scatterers without line-of-sight.
The key similarity between the two models is the uniformly distributed phase of the fading process.
The phase fading and Rayleigh fading models differ in the behaviour of the fading magnitude component,
which is fixed for the former but varies following a Rayleigh distribution for the latter.

\vspace{-0.1cm}

\subsection*{Main Contributions}

	In this paper we establish a DF-based achievable source-channel rate for the MARC with correlated sources and with
side information at the relay and the destination.
 The scheme uses irregular encoding, successive decoding at the relay and backward decoding at the destination.
	We show that for the MARC with correlated sources and side information, irregular encoding yields a higher achievable
source-channel rate than the rate achieved by regular encoding.
	This rate also applies directly to the MABRC.	
	We then derive necessary conditions on the achievable source-channel rate of MARCs (and MABRCs) with correlated sources and side information.

	Next, we consider transmission of correlated sources over fading MARCs and MABRCs with side information at the relay and at the destination.
	Using the necessary conditions on the achievable source-channel rate of the MARC that we derive here, the capacity region of the MARC obtained in \cite[Thm. 9]{Kramer:2005}, 
and the results presented in \cite[Section III.C]{Ron:2010}, we find conditions for correlated sources transmitted over phase fading MARCs with side information, 
under which informational separation is optimal. Optimality conditions are also obtained for Rayleigh fading MARCs with correlated sources and side information. 
Additionally, we find conditions for the optimality of separation for fading MABRCs. {\em This is the first time the optimality of separation is shown for the MARC and the MABRC models}. Note that these models are not degraded in the sense of \cite{Sankar:09}, see also \cite[Remark 33]{Kramer:2005}.
	
	The rest of this paper is organized as follows: in Section \ref{sec:NotationModel}
    the model and notations are presented. In Section \ref{sec:DM_MARC} the separation based achievable source-channel rate is
    presented as well as necessary conditions on the achievable source-channel rate.
	The optimality of separation for correlated sources transmitted over fading Gaussian MARCs is studied in
    Section \ref{sec:sepOpt}. 

\vspace{-0.1cm}

\section{Notations and Model}
\label{sec:NotationModel}

\vspace{-0.1cm}

In the following we use $H(\cdot)$ to denote the entropy of a discrete random variable and
$I(\cdot;\cdot)$ to denote the mutual information between two
random variables, as defined in \cite[ch. 2, ch. 9]{cover-thomas:it-book}.
We denote the set of real numbers with $\setR$, and we use $\setC$ to denote the set of complex numbers.
We denote random variables with upper case letters, e.g., $X$, $Y$, and their realizations with lower case letters
$x$, $y$. A discrete random variable $X$ takes values in a set $\mX$. We use $p_X(x)$ to denote the probability mass function (p.m.f.) of a discrete RV $X$ on $\mX$, and $f_X(x)$ to denote the probability density function (p.d.f.) of a continuous RV $X$ on $\setC$.
For brevity we may omit the subscript $X$ when it is the uppercase version of the sample symbol $x$.
We use $p_{X|Y}(x|y)$ to denote the conditional distribution of $X$ given $Y$.
We use $X^j$ to denote the vector $(X_1,X_2,\dots,X_j)$.
We denote the empty set with $\phi$, and the complement of the set $B$ by $B^c$.
We use $\styp(X)$ to denote the set of $\eps$-strongly typical sequences with respect to distribution
$p_X(x)$ on $\mX$,
 $\mCN(a,\sigma^2)$ to denote a proper, circularly symmetric, complex
Normal distribution with mean $a$ and variance $\sigma^2$, and $\E\{ \cdot \}$ to denote stochastic expectation.

The MARC consists of two transmitters (sources), a receiver (destination) and a relay.
Transmitter $i$ has access to the source sequence $S_i^m$, for $i=1,2$.
The receiver is interested in the lossless reconstruction of the source sequences observed by the two transmitters.
The objective of the relay is to help the receiver decode the source sequences.
It is also assumed that the relay and the receiver have side information correlated with the source sequences.
For the MABRC both the receiver and the relay are interested in a lossless reconstruction of the source sequences.
Figure \ref{fig:MABRCsideInfo} depicts the MABRC with side information setup. 
\begin{figure}[h]
		\vspace{-0.3cm}
    \centering
    \scalebox{0.44}{\includegraphics{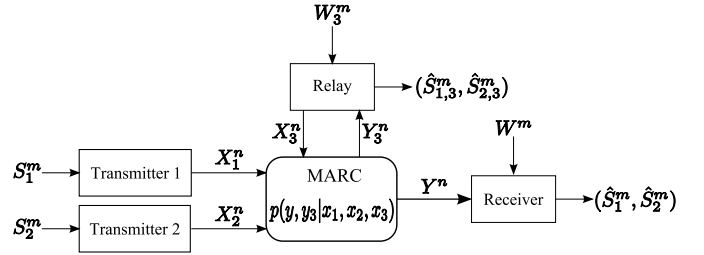}}
    \vspace{-0.44cm}
    \caption{Multiple-access broadcast relay channel with correlated side information. 
    $(\hat{S}^m_{1,3}, \hat{S}^m_{2,3})$ are the reconstructions of $(S^m_{1}, S^m_{2})$ at the relay, and $(\hat{S}^m_{1}, \hat{S}^m_{2})$ are the reconstructions at the destination
    .}
    \label{fig:MABRCsideInfo}
    \vspace{-0.3cm}
\end{figure}

The sources and the side information sequences, $\{ S_{1,k},S_{2,k},W_{k},W_{3,k} \}_{k=1}^{m}$, are
arbitrarily correlated according to a joint distribution $p(s_1,s_2,w,w_3)$ over a
finite alphabet $\mS_1 \times \mS_2 \times \mW \times \mW_3$, and independent across different sample indexes $k$.
All nodes know this joint distribution.

For transmission, a discrete memoryless channel with inputs $X_1, X_2, X_3$ over finite input alphabets $\mX_1,\mX_2,\mX_3$, and outputs $Y, Y_3$ over finite output alphabets $\mY,\mY_3$, is available.
The channel is memoryless in the sense
\begin{equation}
	p(y_{k},y_{3,k}|y^{k-1},y_3^{k-1},x_1^k,x_2^k,x_3^k) = p(y_{k},y_{3,k}|x_{1,k},x_{2,k},x_{3,k}). \nonumber
\label{eq:MARCchanDist}
\end{equation}

An $(m, n)$ source-channel code for the MABRC with correlated side information consists of two encoding functions at the transmitters:
$f_i^{(m,n)} : \mS_i^m \mapsto \mX_i^n, i=1,2$, and two decoding functions at the destination and the relay:
$g^{(m,n)}: \mY^n \times \mW^m \mapsto \mS_1^m \times \mS_2^m, g_3^{(m,n)}: \mY_3^n \times \mW_3^m \mapsto \mS_1^m \times \mS_2^m$.
Finally, there is a causal encoding function at the relay, $x_{3,k} = f_{3,k}^{(m,n)}(y_{3,1}^{k-1},w_{3,1}^m), 1 \leq k \leq n$.
Note that in the MARC scenario the decoding function $g_3^{(m,n)}$ does not exist.
Let $\hat{S}_i^m$ denote the reconstruction of $S_i^m, i=1,2,$ respectively, at the receiver. Let $\hat{S}_{i,3}^m$
denote the reconstruction of $S_i^m, i=1,2,$ respectively, at the relay. The average probability of error, $P_e^{(m,n)}$, of an $(m,n)$
code for the MABRC is defined as
$P_e^{(m,n)} \triangleq  \Pr\big((\hat{S}_1^m,\hat{S}_2^m) \neq (S_1^m,S_2^m) \mbox{ or } (\hat{S}_{1,3}^m,\hat{S}_{2,3}^m) \neq (S_1^m,S_2^m) \big)$, while for the MARC the average probability of error is defined as $P_e^{(m,n)} \triangleq  \Pr\big((\hat{S}_1^m,\hat{S}_2^m) \neq (S_1^m,S_2^m)\big)$.
A source-channel rate $\kappa$ is said to be achievable for the MABRC, if for every $\epsilon > 0$ there exist positive integers
$n_0, m_0$ such that for all $n>n_0, m>m_0,n/m=\kappa$, there exists an $(m,n)$ code for which $P_e^{(m,n)} < \epsilon$.
The same definition applies to the MARC.

	For fading Gaussian MARCs and MABRCs, the received signals at time $k$ at the receiver and at the relay are given by (see Figure \ref{fig:FadeGaussMARC})
	
	\vspace{-0.60cm}
	
	\begin{subequations} \label{eq:FadeGaussMARC}
	\begin{eqnarray}
		Y_{k} &=& H_{11,k} X_{1,k} + H_{21,k} X_{2,k} + H_{31,k} X_{3,k} + Z_{k} \label{eq:FadeGaussMARC_dst} \\
		Y_{3,k} &=& H_{13,k} X_{1,k} + H_{23,k} X_{2,k} + Z_{3,k}, \label{eq:FadeGaussMARC_rly}
	\end{eqnarray}
	\end{subequations}	
	
	\vspace{-0.5cm}
		
	\begin{figure}[h]
    \centering
    \scalebox{0.41}{\includegraphics{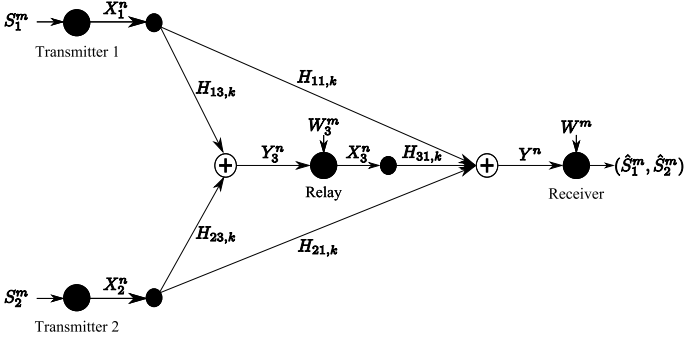}}
    \vspace{-0.6cm}
    \caption{Transmission of correlated sources over the fading Gaussian MARC with side information at the relay and at the destination (additive noises are not depicted).}
    \label{fig:FadeGaussMARC}
    \vspace{-0.3cm}
	\end{figure}	
	
	\noindent $k=1, 2, \dots,n$, where $Z$ and $Z_3$ are independent of each other,
 i.i.d., circularly symmetric, complex Normal RVs, $\mCN(0,1)$.
The channel input signals are subject to per-symbol average power constraints: $\E\{\left|X_i\right|^2 \} \leq P_i, i=1,2,3$.
In the following it is assumed that the destination knows the instantaneous channel coefficients from Transmitter $i, i=1,2$, and from the relay to itself, and the relay knows the instantaneous channel coefficients from both transmitters to itself.
This is referred to as receiver channel state information (Rx-CSI).
	For {\it phase} fading channels the channel coefficients are given by $H_{li,k} = a_{li} e^{j\Theta_{li,k}}$,
where $a_{li} \in \setR$
are constants representing the attenuation and $\Theta_{li,k}$ are uniformly distributed over $[0,2\pi)$, i.i.d., and independent of
each other and of the additive noises $Z_3$ and $Z$. For {\it Rayleigh} fading channels the channel coefficients are given by
$H_{li,k} = a_{li} U_{li,k}$ , $a_{li} \in \setR$ are constants representing the attenuation, and  $U_{li,k}$ are circularly symmetric,
complex Normal RVs, $U_{li,k} \sim \CN(0,1)$, i.i.d.,~and independent of each other and of the additive noises $Z_3$ and~$Z$.

\vspace{-0.0cm}

\section{Source-Channel Coding for Discrete Memoryless MARCs and MABRCs} \label{sec:DM_MARC}

\vspace{-0.0cm}

\subsection{Operational Separation-Based Achievable Rate}

\vspace{-0.1cm}

In this subsection we present an achievable rate for MARCs and MABRCs with correlated sources and side information, based on operational separation.

\begin{theorem}
	\thmlabel{thm:separationCond}
	For DM MARCs and DM MABRCs with relay and receiver side information as defined in Section \ref{sec:NotationModel}, source-channel rate $\kappa$ is achievable if,
	\vspace{-0.15cm}
	\begin{subequations} \label{bnd:sepBased}
	\begin{eqnarray}
		H(S_1|S_2,W_3) &<& \kappa I(X_1;Y_3|X_2, V_1,  X_3) \label{bnd:rly_S1} \\
		H(S_2|S_1,W_3) &<& \kappa I(X_2;Y_3|X_1,  V_2, X_3) \label{bnd:rly_S2} \\
		H(S_1,S_2|W_3) &<& \kappa I(X_1,X_2;Y_3|V_1, V_2, X_3) \label{bnd:rly_S1S2}\\
		H(S_1|S_2,W) &<& \kappa I(X_1,X_3;Y|X_2, V_2) \label{bnd:dst_S1} \\	
		H(S_2|S_1,W) &<& \kappa I(X_2,X_3;Y|X_1, V_1) \label{bnd:dst_S2} \\	
		H(S_1,S_2|W) &<& \kappa I(X_1,X_2,X_3;Y), \label{bnd:dst_S1S2}
	\end{eqnarray}
	\end{subequations}
	
	\vspace{-0.1cm}
	
	\noindent for an input distribution that factors as
	\vspace{-0.2cm}
	\begin{equation*}
		p(s_1,s_2,w_3,w)p(v_1)p(x_1|v_1)p(v_2)p(x_2|v_2)p(x_3|v_1,v_2).
	\label{eq:SepJointDist}
	\end{equation*}
	
	\vspace{-0.25cm}

\end{theorem}

\begin{proof}[$\mspace{-46mu}$ Proof outline]
The achievability is established by using two independent Slepian-Wolf source coding schemes
\cite[Section 14.4]{cover-thomas:it-book}, and a channel coding scheme similar to the one detailed in
\cite[Sections II, III]{Sankar:07}. The channel coding scheme employs a DF code  with irregular
block Markov encoding, successive decoding at the relay,
and backward decoding at the destination. 
Detailed proof is provided in \cite{Murin:IT11}.
\end{proof}

\vspace{-0.1cm}
\subsection{Discussion} \label{subsec:DM_MARC_Discussion}
\vspace{-0.1cm}

	In \Thmref{thm:separationCond}, bounds \eqref{bnd:rly_S1}--\eqref{bnd:rly_S1S2} are constraints for decoding at the relay,
while bounds \eqref{bnd:dst_S1}--\eqref{bnd:dst_S1S2} are  constraints for decoding at the destination.
	The source-channel achievable rate of \Thmref{thm:separationCond} is established by using two different Slepian-Wolf coding schemes: one for the relay and one for the destination. This requires an irregular encoding scheme for the channel code.
	In regular encoding, the codebooks at the source and the relay have the same size, see for example \cite{Sankar:07}. 
    Applying regular encoding to MABRCs with correlated sources and side information leads to merging some of the constraints in \eqref{bnd:sepBased}. 
    In particular \eqref{bnd:rly_S1} and \eqref{bnd:dst_S1} will be combined into the constraint
	
	\vspace{-0.6cm}
	
	\begin{align}
		&\max \big\{ H(S_1|S_2,W_3), H(S_1|S_2,W) \big\} < \nonumber \\
		& \quad \kappa \min \big\{I(X_1;Y_3|X_2, V_1, X_3), I(X_1,X_3;Y|X_2, V_2)\big\}. \nonumber
	\end{align}	
	
	\vspace{-0.1cm}

	For irregular encoding, bounds \eqref{bnd:rly_S1} and \eqref{bnd:dst_S1} need not be combined,
since the transmission rates to the relay and to the destination can be different due to different quality of the side information.
	We conclude that for MABRCs with correlated sources and side information, irregular encoding yields a higher
source-channel achievable rate than the one achieved by regular encoding.
	When the relay and destination have the same side information ($W\!=\!W_3$) then the irregular and regular schemes obtain the
same achievable source-channel rates.

We note that  when using {\em regular encoding} for MARCs, there is a single Slepian-Wolf code, hence, in the scheme used in Thm. 1 it is not required
to recover the source sequences at the relay and the right-hand side (RHS) of the constraints \eqref{bnd:rly_S1}--\eqref{bnd:dst_S1S2} can be combined. For example,
\eqref{bnd:rly_S1} and \eqref{bnd:dst_S1} will be combined into the constraint

	\vspace{-0.6cm}
	
	\begin{align}
		 H(S_1|S_2,W) <  \kappa \min \big\{ & I(X_1;Y_3|X_2, V_1, X_3),\nonumber\\
        & \qquad I(X_1,X_3;Y|X_2, V_2)\big\}. \nonumber
	\end{align}	



\vspace{-0.3cm}

\subsection{Necessary Conditions on the Achievable Source-Channel Rate} \label{sec:outerBounds}

\vspace{-0.15cm}

In this subsection we present necessary conditions on the achievable source-channel rate for MARCs and for MABRCs with correlated
sources and side information at the relay and at the destination.

\vspace{0.1cm}

\begin{proposition}
	\label{thm:OuterGeneral}
	Consider the transmission of arbitrarily correlated sources $S_1$ and $S_2$ over the DM
MARC with relay side information $W_3$ and receiver side information $W$.	
	Any achievable source-channel rate $\kappa$ must satisfy the  constraints:

\vspace{-0.6cm}

\begin{subequations} \label{bnd:outr_general_dst}
\begin{eqnarray}
		H(S_1|S_2,W) &\leq& \kappa I(X_1,X_3;Y|X_2) \label{bnd:outr_general_dst_S1} \\	
		H(S_2|S_1,W) &\leq& \kappa I(X_2,X_3;Y|X_1) \label{bnd:outr_general_dst_S2} \\	
		H(S_1,S_2|W) &\leq& \kappa I(X_1,X_2,X_3;Y) \label{bnd:outr_general_dst_S1S2},
\label{eq:outerGeneral}
\end{eqnarray}
\end{subequations}

\vspace{-0.25cm}

	\noindent for some input distribution $p(x_1,x_2,x_3)$.
\end{proposition}
\vspace{-0.1cm}
\begin{proof}[$\mspace{-42mu}$ Proof]
A detailed proof is provided in \cite{Murin:IT11}.
\end{proof}

\vspace{-0.1cm}
\begin{remark}
	Setting $\mX_2=\mS_2 = \phi$, Proposition \ref{thm:OuterGeneral} specializes to the converse of \cite[Thm. 3.1]{ErkipGunduz:07} for the relay channel.
\end{remark}
\vspace{-0.1cm}

\begin{proposition}
	\label{thm:OuterGeneralMABRC}
	
	Consider the transmission of arbitrarily correlated sources $S_1$ and $S_2$ over the DM MABRC with relay side information $W_3$ and receiver side information $W$.	
	Any achievable source-channel rate $\kappa$ must the satisfy the constraints in \eqref{bnd:outr_general_dst} as well as the following constraints:

\vspace{-0.6cm}

\begin{subequations} \label{bnd:outr_general_rly}
	\begin{eqnarray}
		H(S_1|S_2,W_3) &\leq& \kappa I(X_1;Y_3|X_2, X_3) \label{bnd:outr_general_rly_S1} \\
		H(S_2|S_1,W_3) &\leq& \kappa I(X_2;Y_3|X_1, X_3) \label{bnd:outr_general_rly_S2} \\
		H(S_1,S_2|W_3) &\leq& \kappa I(X_1,X_2;Y_3| X_3), \label{bnd:outr_general_rly_S1S2}
\label{eq:outerRelayGeneral}
\end{eqnarray}	
\end{subequations}
\vspace{-0.65cm}

	\noindent for some input distribution $p(x_1,x_2,x_3)$.

\vspace{-0.1cm}

\end{proposition}	

\begin{proof}[$\mspace{-42mu}$ Proof]
A detailed proof is provided in \cite{Murin:IT11}.
\end{proof}

\vspace{-0.10cm}

\section{Optimality of Source-Channel Separation for Fading Gaussian MARCs and MABRCs}
\label{sec:sepOpt}

\vspace{-0.10cm}

We begin by considering source-channel separation for phase fading Gaussian MARCs \eqref{eq:FadeGaussMARC}. The result is stated in the following theorem.
\vspace{-0.1cm}
	\begin{theorem} \thmlabel{thm:PhaseGaussMARC}
		Consider the transmission of arbitrarily correlated sources $S_1$ and $S_2$ over a phase fading Gaussian
MARC~with receiver side information $W$
 and relay side information $W_3$. 
 Let the channel inputs be subject~to
per-symbol power constraints $\E\{\left|X_i\right|^2 \} \leq P_i, i=1,2,3$, and let the channel coefficients and the channel input powers satisfy
	
	\vspace{-0.60cm}
		
	\begin{subequations} \label{eq:phaseRelayDecConstraints}
	\begin{eqnarray}
		a_{11}^2 P_1 + a_{31}^2 P_3 &\leq& a_{13}^2 P_1 \label{eq:phaseRelayDecConstraints_R1} \\
		a_{21}^2 P_2 + a_{31}^2 P_3 &\leq& a_{23}^2 P_2 \label{eq:phaseRelayDecConstraints_R2} \\
		a_{11}^2 P_1 + a_{21}^2 P_2 + a_{31}^2 P_3 &\leq& a_{13}^2 P_1 + a_{23}^2 P_2.
	\label{eq:phaseRelayDecConstraints_R1R2}
	\end{eqnarray}
	\end{subequations}
	
	\vspace{-0.15cm}

	\noindent A source-channel rate $\kappa$ is achievable if
	
	\vspace{-0.60cm}
	
	\begin{subequations} \label{eq:phaseAchievConstraints}
	\begin{eqnarray}
		H(S_1|S_2,W) &<& \kappa \log_2(1 + a_{11}^2 P_1 + a_{31}^2 P_3) \label{eq:phaseAchievConstraints_S1}	\\
		H(S_2|S_1,W) &<& \kappa \log_2(1 + a_{21}^2 P_2 + a_{31}^2 P_3) \label{eq:phaseAchievConstraints_S2}	\\
		H(S_1,S_2|W) &<& \kappa \log_2(1 + a_{11}^2 P_1 + a_{21}^2 P_2 + a_{31}^2 P_3). \label{eq:phaseAchievConstraints_S1S2}
	\end{eqnarray}
	\end{subequations}

	\vspace{-0.10cm}
		
	\noindent 
		Conversely, if source-channel rate $\kappa$ is achievable, then conditions \eqref{eq:phaseAchievConstraints} are
satisfied with $<$ replaced by $\leq$.
	\end{theorem}

	\vspace{-0.05cm}
	
	\begin{proof}[$\mspace{-42mu}$ Proof]
		See subsections \ref{subsec:achieveProof}, \ref{subsec:converseProof}.
	\end{proof}

\vspace{-0.0cm}	

\begin{remark} \label{rem:achiid}
  The source-channel rate $\kappa$ in \Thmref{thm:PhaseGaussMARC} is achieved by using $X_i \sim \mCN(0,P_i), i \in \{1,2,3 \}$, all i.i.d. and independent of each other, and applying DF at the relay.
\end{remark}


\begin{remark} 
\label{subsec:Fading_Discussion}
	
\vspace{-0.0cm}	
	
		The achievability scheme of \Thmref{thm:PhaseGaussMARC} uses the channel code construction and decoding rules
        detailed in \cite[Section III.C]{Ron:2010}. 		
		The decoding rules detailed in \cite[Section III.C]{Ron:2010} imply that the  destination channel
        decoder does not use any information provided by the  destination source decoder.		
		The only interaction between the source code and channel code is through the bin indices of the transmitted sequences.
        Hence, \Thmref{thm:PhaseGaussMARC} implies that informational separation is optimal.
		
	\end{remark}
	
	Next, we consider sources transmission over Rayleigh fading MARCs.
	
	\begin{theorem} \thmlabel{thm:RayleighGaussMARC}
		Consider the transmission of arbitrarily correlated sources $S_1$ and $S_2$ over a Rayleigh fading Gaussian MARC with
receiver side information $W$
 and relay side information $W_3$. 
Let the channel inputs be subject~to per-symbol power
constraints $\E\{\left|X_i\right|^2 \} \leq P_i, i=1,2,3$, and let the channel coefficients and the channel input powers satisfy
		
		\vspace{-0.5cm}
		
		\begin{subequations} \label{eq:RayleighRelayDecConstraints}
		\begin{align}
			1 + a_{11}^2 P_{1} + a_{31}^2P_{3} & \leq \frac{a_{13}^2 P_1} {e^{ \frac{ 1}{a_{13}^2 P_1}} E_1\left( \frac{1 }{ a_{13}^2  P_{1} }\right)} \\
			1 + a_{21}^2 P_{2} + a_{31}^2P_{3} & \leq \frac{a_{23}^2 P_2} {e^{ \frac{ 1}{a_{23}^2 P_2}} E_1\left( \frac{1 }{ a_{23}^2  P_{2} }\right)} \\
			1 + a_{11}^2 P_{1} + a_{21}^2 P_{2} + a_{31}^2P_{3} & \leq \nonumber
		\end{align}
		
		\vspace{-0.63cm}
		
		\begin{align}
			\frac{a_{23}^2P_2 - a_{13}^2P_1}{\left( e^{ \frac{1}{a_{23}^2P_2}} E_1\left(\frac{1}{a_{23}^2P_2}\right) - e^{\frac{1}{a_{13}^2P_1}}E_1\left(\frac{1}{a_{13}^2P_1}\right) \right)},
		\end{align}
		\end{subequations}
		
		\vspace{-0.1cm}
		
		\noindent where $E_1(x) \triangleq \int_{q=x}^{\infty}{\frac{1}{q}e^{-q}dq}$.  A source-channel rate $\kappa$ is achievable if
		
		\vspace{-0.6cm}
		
		\begin{subequations} \label{eq:RayleighAchievConstraints}
		\begin{align}
			H(S_1|S_2,W) &< \kappa \E_{\tU}\big\{\log_2(1 + a_{11}^2|U_{11}|^2P_1 \nonumber \\ 		
				& \:\:\: \qquad \qquad \qquad \qquad	+ a_{31}^2|U_{31}|^2P_3) \big\} \label{eq:RayleighAchievConstraints_S1}	\\
			H(S_2|S_1,W) &< \kappa \E_{\tU}\big\{\log_2(1 + a_{21}^2|U_{21}|^2P_2 \nonumber \\
			 & \:\:\: \qquad \qquad \qquad \qquad	 + a_{31}^2|U_{31}|^2P_3) \big\} \label{eq:RayleighAchievConstraints_S2}	\\
			H(S_1,S_2|W) &< \kappa \E_{\tU}\big\{\log_2(1 + a_{11}^2|U_{11}|^2P_1 + \nonumber \\
			& \qquad \qquad a_{21}^2|U_{21}|^2P_2  + a_{31}^2|U_{31}|^2P_3) \big\}, \label{eq:RayleighAchievConstraints_S1S2}
		\end{align}
		\end{subequations}
		
		\vspace{-0.15cm}
		
		\noindent where $\tU = \big(U_{11},U_{13},U_{21},U_{23},U_{31}\big)$.
			Conversely, if source-channel rate $\kappa$ is achievable, then conditions \eqref{eq:RayleighAchievConstraints} are
satisfied with $<$ replaced by $\leq$.
	\end{theorem}
	
	\begin{proof}[$\mspace{-42mu}$ Proof]
		The proof uses \cite[Corollary B.1]{Ron:2010} and follows similar arguments to those in the proof of \Thmref{thm:PhaseGaussMARC}.
	\end{proof}	
	
\begin{remark}
  The source-channel rate $\kappa$ in \Thmref{thm:RayleighGaussMARC} is achieved by using $X_i \sim \mCN(0,P_i), i \in \{1,2,3 \}$, all i.i.d. and independent of each other, and applying DF at the relay.
\end{remark}

\vspace{-0.2cm}
	
\subsection{Achievability Proof of \Thmref{thm:PhaseGaussMARC}} \label{subsec:achieveProof}
\vspace{-0.1cm}
\subsubsection{Code construction}
For $i\!=\!1,2$, assign every $\svec_i \in \mS_i^m$~to one of $2^{mR_i}$ bins independently according to a uniform
distribution on $\{1,2,\dots,2^{mR_i}\}$. Denote these assignments by~$f_i$. A channel code based on DF with rates
$\chR_1$ and $\chR_2$, and a blocklength $n$, is constructed as detailed in \cite[Section III.C]{Ron:2010}.
Transmitter $i$ has $2^{n\chR_i}$ messages, thus we require $\kappa \chR_i= R_i$.

\subsubsection{Encoding} Consider a source sequence of length $Bm$, $s_i^{Bm} \in \mS^{Bm}_i, i=1,2$.
Partition this sequence into $B$ length-$m$ subsequences, $\svec_{i,b}, b=1,2,\dots,B$. 
Similarly partition the
side information sequence $w^{Bm} \in \mW^{Bm}$,
into $B$ length-$m$ subsequences.
We transmit a total of $Bm$ source samples over $B+1$ blocks of $n$ channel uses each. If we set $n=\kappa m$,
by increasing $B$ we obtain a source-channel rate $(B+1)n/Bm \rightarrow n/m = \kappa$ as $B \rightarrow \infty$.

In block $b, b=1,2,\dots,B$, source terminal $i, i=1,2$, observes $\svec_{i,b}$ and finds its corresponding
bin index $\msg_{i,b} \in \{1,2,\dots,2^{mR_i}\}$. Each transmitter sends its corresponding bin index using the
channel code described in \cite[Section III.C]{Ron:2010}.
			
Encoding at the relay: Assume that at time $b$ the relay knows 
$(\msg_{1,b-1},\msg_{2,b-1})$. The relay sends these bin indices using the channel code described in
\cite[Section III.C]{Ron:2010}.	
	
\subsubsection{Decoding and error probability} Conditions \eqref{eq:phaseRelayDecConstraints} imply that the
achievable channel rate region for decoding at the relay contains the achievable channel rate region for
decoding at the destination. Hence, reliable decoding of the channel code at the destination implies reliable
decoding of the channel code at the relay. When the channel coefficients and the channel input powers satisfy
conditions \eqref{eq:phaseRelayDecConstraints}, the RHS of constraints \eqref{eq:phaseAchievConstraints} is the capacity
region of the phase fading Gaussian MARC, (see \cite[Thm. 9]{Kramer:2005}). Hence, the transmitted bin indices $\{u_{1,b},u_{2,b} \}_{b=1}^B$ can be
reliably decoded at the destination if

\vspace{-0.56cm}

\begin{subequations} \label{eq:phaseAchieveChanRates}
\begin{eqnarray}
	R_1 &\leq& \kappa \log_2(1 + a_{11}^2 P_1 + a_{31}^2 P_3) \\
	R_2 &\leq& \kappa \log_2(1 + a_{21}^2 P_2 + a_{31}^2 P_3) \\
	R_1 + R_2 &\leq& \kappa \log_2(1 + a_{11}^2 P_1 + a_{21}^2 P_2 + a_{31}^2 P_3).
\end{eqnarray}
\end{subequations}

	\vspace{-0.17cm}

	\textit{Decoding the sources at the destination:} The decoded bin indices, denoted
$\tilde{\msg}_{i,b}, i\!=\!1,2, b\!=\!1,2,\dots,B$, are then given to
the source decoder at the destination. Using the bin indices and the side information $\wvec_{b}$, the
source decoder at the destination estimates $\svec_{1,b}, \svec_{2,b}$. More precisely, given the bin
indices $\tilde{\msg}_{1,b}, \tilde{\msg}_{2,b}$, it declares
$(\tilde{\svec}_{1,b},\tilde{\svec}_{2,b})$ to be the decoded sequences if it is the unique pair of sequences
that satisfies $f_1(\tsvec_{1,b})= \tilde{\msg}_{1,b}, f_2(\tsvec_{2,b})= \tilde{\msg}_{2,b}$ and
$(\tsvec_{1,b},\tsvec_{2,b},\wvec_{b}) \in \stypm(S_1,S_2,W)$.
From the Slepian-Wolf theorem \cite[Thm 14.4.1]{cover-thomas:it-book}, $(\svec_{1,b},\svec_{2,b})$
can be reliably decoded at the destination if
	
	\vspace{-0.60cm}

\begin{subequations} \label{eq:phaseAchieveDestSourceRates}
\begin{eqnarray}
	H(S_1|S_2,W) &<& R_1 \\
	H(S_2|S_1,W) &<& R_2 \\
	H(S_1,S_2|W) &<& R_1 + R_2.
\end{eqnarray}
\end{subequations}

	\vspace{-0.18cm}
			
	Combining conditions \eqref{eq:phaseAchieveChanRates} and \eqref{eq:phaseAchieveDestSourceRates}
yields \eqref{eq:phaseAchievConstraints} and completes the achievability proof.
	
\vspace{-0.15cm}

\subsection{Converse Proof of \Thmref{thm:PhaseGaussMARC}} \label{subsec:converseProof}
\vspace{-0.1cm}

Consider the necessary conditions of Proposition \ref{thm:OuterGeneral}. From \cite[Thm. 8]{Kramer:2005} it follows that for phase fading
with Rx-CSI, the mutual information
expressions on the RHS of \eqref{bnd:outr_general_dst} are simultaneously maximized by $X_1,X_2,X_3$ independent, zero-mean
complex Normal, $X_i \sim \mCN(0,P_i), i=1,2,3$ yielding the same expressions  as in \eqref{eq:phaseAchievConstraints}.
	Therefore, for phase fading MARCs, when conditions \eqref{eq:phaseRelayDecConstraints} 
hold, the conditions in \eqref{eq:phaseAchievConstraints} coincide with the necessary conditions of Proposition \ref{thm:OuterGeneral},
and the conditions in \eqref{eq:phaseAchievConstraints} are satisfied with $\leq$ instead of $<$.

\vspace{-0.15cm}

\subsection{Fading MABRCs}
\vspace{-0.1cm}
The optimality of informational separation can be established for MABRCs 
using the results for the MARC  with additional constraints, as indicated in the following theorem:
\begin{theorem}
\label{thm:separation-MABRCs}
    For phase fading MABRCs for which the conditions in \eqref{eq:phaseRelayDecConstraints} hold together with
%
%
%
%
	\vspace{-0.20cm}
	\begin{subequations} \label{eq:phaseEntropyConstraints}
	\begin{eqnarray}
		H(S_1|S_2,W_3) &\leq&  H(S_1|S_2,W) \\
		H(S_2|S_1,W_3) &\leq&  H(S_2|S_1,W) \\
		H(S_1,S_2|W_3) &\leq&  H(S_1,S_2|W),
	\end{eqnarray}
	\end{subequations}

	\vspace{-0.25cm}

    \noindent
    the maximum achievable source-channel rate $\kappa$ satisfies \eqref{eq:phaseAchievConstraints}.
    The same statement holds for Rayleigh fading MABRCs with \eqref{eq:RayleighRelayDecConstraints} replacing \eqref{eq:phaseRelayDecConstraints} and \eqref{eq:RayleighAchievConstraints} replacing \eqref{eq:phaseAchievConstraints}.
\end{theorem}	
	\vspace{-0.1cm}

	\begin{remark}
		Conditions \eqref{eq:phaseEntropyConstraints} imply that for the scenario described in \Thmref{thm:PhaseGaussMARC}
regular encoding and irregular encoding yield the same source-channel achievable rates (see discussion in Section
\ref{subsec:DM_MARC_Discussion}), hence the channel code construction of \cite[Section III.C]{Ron:2010} can be used.
	\end{remark}

\begin{remark}
    The proof for MABRCs differs from the achievability proof of section \ref{subsec:achieveProof} only due to decoding the source 
    sequences at the relay. This decoding follows similar arguments to the decoding of
the sources at the destination. Conditions \eqref{eq:phaseEntropyConstraints} imply that reliable decoding of the sources
at the destination guarantees reliable decoding of the sources at the relay, since the relay achievable source rate region
contains the destination achievable source rate region.
\end{remark}

\end{document}